\pdfoutput=1
\documentclass[cleveref]{lipics-v2021}
\usepackage{amsmath}
\usepackage{amssymb}
\usepackage{amsthm}
\usepackage{xcolor}
\usepackage{hyperref}
\usepackage{makecell}
\usepackage{comment}
\usepackage{bm}
\usepackage{enumerate}
\usepackage{booktabs}
\usepackage[commentsnumbered,vlined]{algorithm2e}
\DontPrintSemicolon

\SetCommentSty{commentFont}
\SetKwProg{algo}{Algorithm}{:}{}
\usepackage{etoolbox}
\patchcmd{\SetProgSty}{ArgSty}{ProgSty}{}{}
\SetProgSty{upshape}
\def\op#1{\ensuremath{\mathsf{#1}}}
\def\Sdec{S_{\mathrm{dec}}}
\def\QN{Q_{\mathrm{next}}}

\def\bpfrac#1#2{\Big(\frac{#1}{#2}\Big)}

\newcommand{\gkk}[1]{{#1}}
\newcommand{\gk}[1]{{#1}}

\newcommand{\Bin}{\mathrm{Bin}}

\title{\gkk{Better space-time-robustness trade-offs for set reconciliation}}

 \author{Djamal Belazzougui}{CAPA, DTISI, Centre de Recherche sur l'Information Scientifique et Technique, Algiers, Algeria}{djamal.belazzougui@gmail.com}{}{}
 \author{Gregory Kucherov}{LIGM, CNRS, Universit\'e Gustave Eiffel, Marne-la-Vall\'ee, France}{Gregory.Kucherov@univ-eiffel.fr}{https://orcid.org/0000-0001-5899-5424}{}
 \author{Stefan Walzer}{Karlsruhe Institute of Technology, Germany}{stefan.walzer@kit.edu}{https://orcid.org/0000-0002-6477-0106}{}
 
 \authorrunning{D.~Belazzougui and G.~Kucherov and S.~Walzer}
 
 \Copyright{Djamal Belazzougui and Gregory Kucherov and Stefan Walzer}
 
\ccsdesc{Theory of computation~Design and analysis of algorithms}
\ccsdesc{Theory of computation~Randomness, geometry and discrete structures}
\ccsdesc{Theory of computation~Sketching and sampling}
\ccsdesc{Theory of computation~Error-correcting codes}

\keywords{data structures, hashing, set reconciliation, invertible Bloom lookup tables, random hypergraphs, BCH codes} 

\hideLIPIcs
\nolinenumbers

\date{}

\def\lookspure{\mathtt{looksPure}}

\begin{document}
	
	\maketitle
	
\begin{abstract}
\gk{
	We consider the problem of reconstructing the symmetric difference between similar sets from their representations (sketches) of size linear in the number of differences. Exact solutions to this problem are based on error-correcting coding techniques and suffer from a large decoding time. Existing probabilistic solutions based on Invertible Bloom Lookup Tables (IBLTs) are time-efficient but offer insufficient success guarantees for many applications. Here we propose a tunable trade-off between the two approaches combining the efficiency of IBLTs with exponentially decreasing failure probability. The proof relies on a refined analysis of IBLTs proposed in (B{\ae}k Tejs Houen et al. SOSA 2023) which has an independent interest. We also propose a modification of our algorithm that enables telling apart the elements of each set in the symmetric difference. 
}
\end{abstract}
	
	\section{Introduction}
	
	The problem of \textit{set reconciliation} (or \textit{database reconciliation} \cite{goodrich2011invertible}) lies at the intersection of several lines of algorithmic research. One intuitive way to define it is in terms of communication protocols: assume that two parties Alice and Bob hold similar sets $S$ and $T$, respectively, that they seek to reconcile (i.e.\ identify differences) using low communication overhead. Rather than exchanging representations of entire sets, they exchange their small \textit{sketches} holding enough information to identify differences of $S$ and $T$ provided that their number does not exceed a pre-defined parameter~$D$. 
	
	In this paper, we study the version of the problem where the differences $S\triangle T := (S\setminus T) \cup (T\setminus S)$ should be recovered from the sketches of $S$ and $T$ rather than from the sketch of one of them and the entire other set. This opens a way to the setting, relevant to some applications, when only sketches of sets are stored in a database, and for each pair of sets, items proper to one of them can be retrieved from the corresponding sketches. As an example, consider a vast database of highly similar genomic sequences, such as those of SARS-Cov-2 genomes\footnote{More than 16M SARS-Cov-2 are available in GISAID database at the time of writing, which typically differ one from another by just a few characters, over about 30,000 characters of length}, that are stored in the form of sketches, rather than raw sequences, still supporting retrieval of differences between any two of them. Beyond this example, set reconciliation problem occurs in many other scenarios in distributed systems, where information needs to be synchronized between computational units. Those include blockchain systems, 
	P2P systems, 
	ad-hoc networks, synchronizing information across devices or data centers, and others. 
	
	The set reconciliation problem is also raised in the streaming framework. It is a common and natural requirement that the sketch can be efficiently updated when a new key is added to or deleted from the set. For example, \cite{EG10} considers the \textit{straggler identification} problem defined on a stream of insertions and deletions of keys modeling a stream of people entering and leaving a building. In this model, the sketch summarizes information about people currently in the building and if their number does not exceed $D$, is capable of listing them. 
		
	\subsection{Prior work}
	\label{sec:prior_work}
	Solutions to set reconciliation can be categorized into probabilistic and exact ones (although there exists an interplay between these two classes of algorithms) \cite{MorganPhdthesis18}. Probabilistic algorithms are based on Invertible Bloom Lookup Tables (IBLTs), also known as Invertible Bloom Filters, based on random hash functions. Originally proposed in \cite{EG10} where they were applied to the problem of \textit{subset reconciliation} (set reconciliation assuming $S\subseteq T$), 
	they (in a slightly different version) were more systematically studied in \cite{goodrich2011invertible}. Earlier a related data structure, called \textit{$k$-set data structure}, was proposed in \cite{GANGULY2007211}. Set reconciliation by computing a difference of IBLTs (called \textit{difference digest}) was studied in \cite{EGUV2011}. 
	Work \cite{mitzenmacher2018simple} further applies IBLTs to \textit{multi-party} set reconciliation. The idea of ``subtracting'' IBLTs has also been applied in \cite{FLS-EUROCRYPT22}, where the emphasis is to store in an IBLT carefully defined hash values obtained using cryptography techniques. Very recent work \cite{HPW22} studies a simplified version of IBLTs that also applies to set reconciliation.
	We rely on this construction in the present paper. 
	
	An exact solution of 
	the set reconciliation problem uses algebraic techniques, in particular \textit{error-correcting codes} \cite{EGUV2011,DBLP:conf/allerton/Cheraghchi11,cheraghchi2019simple}. 
	Work \cite{1226606} proposes an exact solution based on characteristic polynomials to both subset and set reconciliation problems using $(D+1)\log U$ bits of space 
	(i.e.\ essentially the space needed to store $D$ differences). \cite{GANGULY200827} studies the straggler identification problem with multiplicities and proposes an $\mathcal{O}(D\log(Um))$-bit solution based on polynomials over finite field, where each key occurs at most $m$ times. A somewhat similar solution to subset reconciliation was proposed in \cite{EG10}, based on Newton polynomials.  \cite{DBLP:conf/allerton/Cheraghchi11} surveys code-theoretic techniques for space recovery many of which apply to set reconciliation as well. Set reconciliation with BCH codes has been implemented in Minsketch software \cite{minisketch}, leveraging an efficient syndrome decoding algorithm for BCH codes \cite{DBLP:journals/siamcomp/DodisORS08}. 
	\cite{cheraghchi2019simple} proposes a method for sparse recovery based on expander codes; the construction of \cite{cheraghchi2019simple} can be interpreted as a table somewhat similar to IBLT but with a different decoding mechanism. Recent papers  \cite{DBLP:journals/corr/abs-2212-13812,DBLP:conf/isit/BarLevMERY23} apply algebraic techniques to construct specific ``hash functions'' that guarantee successful decoding for IBLTs. The general downside of exact solutions is a larger decoding time, typically growing at least quadratically in $D$ and relying on finite field arithmetic. Note also that the relationship between codes and set reconciliation is two-way: \cite{mitzenmacher2012biff} proposes a construction of codes based on set reconciliation with IBLTs.

	\subsection{Our contribution: overview}	
	The IBLT solution to set reconciliation is very efficient but provides a poor success guarantee. For example, if we have to compute similarity joins by performing	 all pairwise reconciliations of the objects in our database,  then a significant fraction of comparisons may fail. On the other hand, techniques based on error-correcting codes are exact but have a high cost of decoding. 
	Here we propose a family of solutions to set reconciliation offering a trade-off between these two approaches, that is having significantly smaller error rate compared to IBLT at the price of an asymptotically vanishing increase of space and time bounds. 

	Our solution is based on the space-efficient IBLT from \cite{HPW22} complemented by an additional \textit{stash} data structure supporting recovery in case of failure of the main IBLT. \gk{The analysis of \cite{HPW22} provides no guarantee in the case when the IBLT decoding fails. Therefore, here we enhance the analysis of \cite{HPW22} and prove a probability bound on the number of missing and extraneous elements reported by the data structure, in the event that the correct decoding fails. This is a key argument of our construction and its proof constitutes the main technical contribution of this paper.}
	
	The algorithm we obtain is parametrized and provides a tunable trade-off between the failure probability and space and time consumption. 
	In the table below, we provide an overview of how our solution relates to the existing approaches surveyed in Section~\ref*{sec:prior_work}. 
	\begin{table}[h!]
		\begin{center}
			{\small
			\begin{tabular}{c@{}c@{}c@{\ }c@{}cc} 
				\toprule
				\bfseries Method&\makecell{\bfseries Sketch size (bits)}&\bfseries \makecell{Insertion\\time}&\bfseries \makecell{Decoding\\time}&\bfseries \makecell{Failure\\probability}\\
				\midrule
				IBLT \cite{HPW22} &$(c_3+\epsilon)D\log U$&$\mathcal{O}(1)$ &$\mathcal{O}(D)$ &$\mathcal{O}(D^{-1})$\\ 
				\makecell{IBLT with $t$-bit \\hashsum field} & $(c_3+\epsilon)D(\log U+t)$ & ~~$\mathcal{O}(1+t/\gkk{\log U})$ & $\mathcal{O}(D(1+t/\gkk{\log U}))$~~ & $\min\left(\mathcal{O}(D^{-1}),\mathcal{O}(D/2^t)\right)$\\
				BCH &$D\log U$&\gkk{$\mathcal{O}(D\log U)$} &\gkk{$\mathcal{O}(D^2\log^2 U)$} &$0$\\ 
				expander code \cite{cheraghchi2019simple} & $\mathcal{O}(D\log^2 U)$ & $\mathcal{O}(\log U)$ & $\mathcal{O}(D\log U)$ & 0 \\ 
				\midrule
				this paper & $(c_3+\epsilon)D\log U+r\log U$ & \gkk{$\mathcal{O}(r\log U)$} & \gkk{$\mathcal{O}(D)$} & $2^{-\Omega(r)}$ \\ 
				\bottomrule
			\end{tabular}
		}
		\end{center}
		\caption{Comparison of main techniques of set reconciliation to the results of the present work. $r$ is a parameter assumed to verify \gkk{$r=\min\left(\mathcal{O}(D/\log^2 U),\mathcal{O}(\log U)\right)$}. For our algorithm, decoding time given is expected.}
		\label{table:1}
	\end{table}

	The paper is organized as follows. In Section~\ref{sec:IBLT}, we introduce our version of IBLT, which is a slightly modified set sketch from \cite{HPW22}. This data structure is very space-efficient, however its compactness has a price: the decoding process can ``go wrong'' triggering undesirable \textit{anomalous} steps \cite{HPW22}. 
	In Section~\ref{sec:anomalies}, we 
	provide a refined analysis of anomalies and prove the following fundamental property: if the sketch stores a set $S$, the decoding produces a set $S_{dec}$ with $|S \triangle  S_{dec}| \leq  r$, with a failure probability of $2^{-\Omega (r)}$. 
	Section~\ref{sec:err_correct} introduces a solution to set reconciliation based on BCH error-correcting codes. 
	Relying on that result, in Section~\ref{sec:IBLT_with_stash}, we extend the IBLT with a backup data structure (stash) to obtain an efficient solution to set reconciliation. The algorithm resorts to the stash when the decoding with the main IBLT is not completed, leaving out a small number of keys. The stash is defined using error-correcting codes. In Section~\ref{sec:distinguishing}, we discuss how to extend our algorithm in order to identify the original set of the keys of the set difference. 
	
	\section{Set reconciliation with Invertible Bloom Lookup Tables}
	\label{sec:IBLT}
	\subparagraph{Definition of IBLT}
	\label{sec:def_IBLT}
	An Invertible Bloom Lookup Table (IBLT) is an array $A[1:n]$ equipped with $k$ random hash functions $h_1,\ldots,h_k:\mathcal{U}\rightarrow [n]$, from the key universe $\mathcal{U}$  to $[n]=\{1,\ldots,n\}$. We assume $h_i$'s to be fully \gkk{random} and denote $h(x)=\{h_1(x),\ldots,h_k(x)\}$. 
	
	Several variants of IBLTs have been considered, which differ in how entries $A[i]$ are configured. In all of them, $A[i]$ contains a 
	keysum field that holds either bitwise XOR \cite{goodrich2011invertible,EGUV2011,mitzenmacher2018simple,HPW22} or arithmetic sum \cite{EG10,goodrich2011invertible} of all current keys hashed to $i$. XOR provides a more elegant and space-efficient option whereas arithmetic sum becomes necessary when multi-sets are considered. In this work, the keysum field is defined with XOR and is set to contain $\log U$ bits ($U=|\mathcal{U}|$ ). 
	
	Besides the keysum, $A[i]$ can include a hashsum field that holds a hash sum of the keys in $A[i]$ and/or a counter field that tracks the current number of keys in $A[i]$. These are used for enforcing proper functioning and integrity of the data structure at the price of using additional space. 
	In this work, we use the most compact IBLT configuration with three hash functions ($k=3$) and $A[i]$ storing the keysum field only. This variant was introduced and analysed in recent paper \cite{HPW22} that we rely on in this work. Later in Section~\ref{sec:distinguishing} we consider an extension introducing a restricted counter field. 
	
\begin{figure}
	\begin{minipage}{0.55\textwidth}
		\begin{algorithm}[H]
			\algo{\texttt{initialise}}{
				$A[1:n] = (0,\dots ,0)$\;
			}
		\end{algorithm}
		\begin{algorithm}[H]
			\algo{\texttt{toggle}($x$)}{
				\For{$i \in  h(x)$}{
					$A[i] \leftarrow  A[i] \oplus  x$\;
				}
			}
		\end{algorithm}
		\begin{algorithm}[H]
			\algo{\texttt{merge}($A[1:n],A'[1:n]$)}{
				\For{$i \in  [n]$}{
					$B[i] \leftarrow  A[i] \oplus  A'[i]$\;
				}
				\Return $B$
			}
		\end{algorithm}
		\begin{algorithm}[H]
			\algo{\texttt{looksPure}($i \in  [n]$)}{
				\Return $A[i] \neq  0 \wedge  i \in  h(A[i])$\;
			}
		\end{algorithm}
	\end{minipage}\hspace{-6em}
	\begin{minipage}{0.65\textwidth}
		\begin{algorithm}[H]
            \SetKw{Not}{not}
            \SetKw{ManualIf}{if}
            \algo{\texttt{decode}}{
                $\Sdec \leftarrow  \varnothing $\;
                $Q \leftarrow  \{i \in  [n] \mid \op{looksPure}(i)\}$\;
				$(t,t_{\max}) \leftarrow  (0,2n)$ \tcp{time limit $2n$}
                \While{$Q \neq  \varnothing $}{
                    $\QN \leftarrow  \varnothing $\;
                    \For{$i \in  Q$ \ManualIf \op{looksPure}$(i)$}{
                        $x \leftarrow  A[i]$ \tcp{\emph{detected} key $x$} 
                        $\op{toggle}(x)$ \;
                        $\Sdec \leftarrow  \Sdec \triangle  \{x\}$\;  
                        $\QN \leftarrow  \QN \cup  \{i \in  h(x) \mid \op{looksPure}(i)\}$\;
                        $t \leftarrow  t + 1$\;
                        \If{$t \geq  t_{\max}$}{
                        	\Return $\Sdec$\;
                        }
                    }
                    $Q \leftarrow  \QN$
                }
                \Return $\Sdec$
            }
        \end{algorithm}
	\end{minipage}
	\caption{IBLT implementation from \cite{HPW22} with added time limit in decode.}
	\label{fig:algorithms}
\end{figure}
	
Our IBLT implementation is defined in \cref{fig:algorithms}.  Initially, all entries $A[i]$ are set to zero with $\mathtt{initialise}$.
Both insertion and deletion of a key $x$ is done by performing $\mathtt{toggle}(x)$. 

	Decoding the keys stored in an IBLT resembles the \textit{peeling} process of the $k$-hypergraph, where $n$ IBLT entries correspond to hypergraph vertices and each stored key corresponds to a hyperedge defined as the set of entries the key is hashed to. Peeling a hypergraph iterates the following operation: for any vertex of degree 1,  delete the incident hyperedge. If the peeling results in the empty graph with no hyperedges, the input graph is called \textit{peelable}; otherwise the process yields the largest subhypergraph with every vertex of degree at least 2, which is called the \textit{$2$-core} (hereafter simply \textit{core}). 
	It is known that peelability of random $k$-hypergraphs ($k\geq 3$) with $n$ vertices and $m$ hyperedges undergoes a phase transition: if $n > c_k m$, then the hypergraph is peelable with high probability (hereafter, \textit{whp}), whereas if $n < c_k m$, the hypergraph is not peelable whp. Here constant $c_k$ is the peelability threshold, in particular 
	$c_3 = 1.22179\dots $  is the smallest of $c_k$ for $k\geq 3$. The following results is shown in \cite{goodrich2011invertible}. 
	\begin{theorem}
		\label{thm:core_proba}
		Whenever $n > c_k m$, a random $k$-hypergraph is peelable except with probability $\mathcal{O}(1/n^{k-2})$. 
	\end{theorem}
%
	\subparagraph{Set reconciliation using difference IBLT}
	\label{sec:difference_IBLT}
	Consider two sets $S$ and $T$ stored in IBLTs $A_S[1:n]$ and $A_T[1:n]$ respectively, using the same hash functions. If $S$ and $T$ have a bounded symmetric difference, that is $|S\Delta  T|\leq D$, and $n > c_k D$, then the keys of $S\Delta  T$ can be recovered whp from the \textit{difference IBLT} $A_{ST}[1:n]=\mathtt{merge}(A_S,A_T)$. It is immediate to see that common keys of $S$ and $T$ are canceled out and $A_{ST}$ stores exactly the keys $S\Delta  T$. Thus, the difference IBLT can solve set reconciliation except with probability $\mathcal{O}(1/D)$. 
	
	Note that $A_{ST}$ does not allow distinguishing keys of $S \setminus T$ and $T \setminus S$; below in Section~\ref{sec:distinguishing} we will extend the construction to make this possible. Note also that for the reconciliation to be possible, the size $n$ of IBLTs only depends on the size of the symmetric difference and does not depend on the sizes of $S$ and $T$. 
	

	\section{Improved Guarantees for IBLTs}
	\label{sec:main_iblt}
	
	Consider the IBLT implementation as given in \cref{fig:algorithms}. 
	It uses only the keysum field in its buckets and was studied, modulo slight changes stated below, in \cite{HPW22} under the name ``simple set sketch''. 
	In this section we enhance the analysis given in \cite{HPW22} (reproduced as (i) in Theorem~\ref{thm:main-inner} below) by also probabilistically bounding the magnitude of the error of decode in its failure cases (given as (ii)). Note that we restrict our attention to three hash functions ($k=3$) as this yields the best threshold value.\footnote{In \cite{HPW22} any $k \geq  3$ was considered and our analysis here could likewise be extended.}
	
	\begin{theorem}
		\label{thm:main-inner}
		Let $\varepsilon  > 0$ and $n>(c_3+\varepsilon )|S|$. Let $\Sdec$ be the set returned by decode.
		\begin{enumerate}[\upshape(i)]
			\item  $\Pr[\Sdec = S] = 1-\widetilde{\mathcal{O}}(1/n)$.
			\item  For any $r = o(n)$ we have $\Pr[|S \triangle  \Sdec| > r] = 2^{-\Omega (r)}$.
		\end{enumerate}
	\end{theorem}
	\subparagraph{Changes made to decode.}
	Our version of decode differs from that in \cite{HPW22} in two minor ways. First, we have introduced a \emph{time limit} of $2n$ to the number of peeling steps that are performed and return the current state of $\Sdec$ when the time limit is reached\footnote{In practice a dynamic time limit could be considered. For instance: Abort if a round of decode only toggles keys that have already been toggled in a previous round.}. The result from \cite{HPW22} still applies with this change, since the authors proved that decode terminates in $n + \mathrm{poly\,log}(n)$ peeling steps with $\Sdec = S$ with probability $1-\mathcal{O}(1/n)$.
	Second, we insist that $Q$ and $\QN$ are implemented as set data structures, e.g. using hash tables, whereas \cite{HPW22} pointed out that $Q$ and $\QN$ can be implemented as multiset data structures, e.g.\ using FIFO queues.%
	\footnote{At no point did the analysis from \cite{HPW22} actually \emph{rely} on $Q$ or $\QN$ being multisets. The motivation was merely that FIFO queues are the more lightweight data structure.
	We revert this simplification here because an argument from \cite{HPW22} showing that the impact of duplicated elements is negligible no longer applies. In fact, with a FIFO implementation of $Q$ and $\QN$ we could see $2^R$ copies of the same element in $Q$ in the $R$th iteration of the while loop with small but non-negligible probability $n^{-\mathcal{O}(1)}$.}
	
	It should be clear that our version of decode can be implemented such that it has running time $\mathcal{O}(n)$, both in expectation and with high probability.
	
	\subsection{The issue of anomalies}

	Decode checks if the $i$th \emph{bucket} of $A$ contains a single key
	using the function $\lookspure(i)$ that checks if $h_j(A[i])=i$ for some $j\in\{1,2,3\}$. Importantly, $\lookspure(i)$ may produce a false positive result when $A[i] = x_1 \oplus  x_2 \oplus  \dots  \oplus  x_{\ell -1}$ is the sum of several keys hashing to $i$ and $y = A[i]$ also happens to hash to $i$, i.e.\ $h_\sigma (y) = i$ for some $\sigma \in\{1,2,3\}$. Such an occurence $\{x_1,\dots ,x_{\ell -1},y\}$ is called an \textit{anomaly} (of size $\ell $) in \cite{HPW22} and causes the key $y$ to be toggled and thus inserted into the data structure.
	Anomalies corrupt the decoding process and
	may cause decode to fail or even to run in an infinite loop if no time limit is set \cite{HPW22}.  


	Dealing with anomalies is a challenging issue. In \cite{HPW22}, 
	it was proved that certain kinds of anomalies and certain interactions of several anomalies do not occur except with probability $\tilde{O}(1/n)$. Since $\mathcal{O}(1/n)$ is the probability for the existence of a non-empty core (see Theorem~\ref{thm:core_proba} for $k=3$), which also causes decoding to fail, anomalies do not increase the failure probability significantly.
	
	In this section we will go further, proving that \emph{even in the unlikely cases where decoding produces an incorrect result}, e.g.\ due to problematic anomalies, the magnitude of the error is likely to be small.
	
	

	
	\subsection{Proof Outline}
	\label{sec:anomalies}
	
	We quantify the effect of anomalies by considering two sets. First, the set $F$ of all \emph{foreign keys}, which are keys from $\mathcal{U} \setminus S$ that are added to the sketch at least once \gkk{during the decoding process}. Second, the set $N$ of all centres of certain anomalies we call \emph{internal anomalies} (defined in \cref{sec:internal-anomalies}). Our proof of Theorem~\ref{thm:main-inner} rests on the following arguments.
	\begin{itemize}
		\item  In \cref{sec:foreign-keys} we bound the size of $F$. Intuitively, the heuristic \op{looksPure} is used $\mathcal{O}(n)$ times and fails each time with probability $\mathcal{O}(1/n)$. Each failure can cause one foreign key to be added, suggesting $\mathbb{E}[|F|] = \mathcal{O}(1)$. The precise argument, including a concentration bound for $|F|$, is more subtle since failures do not happen independently.
		\item  In \cref{sec:internal-anomalies} we import a tail bound on $|N|$ from \cite{HPW22}.
		\item  In \cref{sec:early-peeling} we reframe the decoding process as the well-understood \emph{peeling process} on a random hypergraph. Since peeling follows a local rule, the corruption due to anomalies spreads locally as well. This already implies that the majority of keys are correctly decoded if anomalies are rare.
		\item  In \cref{sec:late-peeling} we prove a claim regarding the maximum hyperedge density of any subhypergraph of a given random hypergraph. Its intuitive role is to show that the size of the queue $Q$ in late rounds of decode would always be, in the absense of anomalies, at least linear in the number of remaining keys. This implies that the anomalies only have the potential to stall progress in a significant way if they affect a number of buckets that is linear in the number of remaining keys.
		\item  In \cref{sec:main-proof} we stitch everything together: The effect of anomalies as quantified by $|F|$ and $|N|$ is probabilistically bounded. Peeling is likely to proceed largely unhindered in its early rounds. In its later rounds it proceeds until the number of remaining keys is linear in $|F|+|N|$.
	\end{itemize}
	
	\subsection{Bounding Foreign Keys}
	\label{sec:foreign-keys}

	Let $y \in  \mathcal{U} \setminus S$ be any foreign key. In order for $y$ to be added to the IBLT, $y$ must occur as the sum of several keys in one of the buckets $h_1(y),h_2(y),h_3(y)$. Since $h_1(y),h_2(y),h_3(y)$ are uniformly random in $[n]$ and independent of the initial state of the data structure, decode would have to effectively ``guess'' a hash of $y$. We argue that decode is, in this sense, a player in a corresponding guessing game that, even if played optimally, rarely allows a player to guess a lot of hashes correctly.
	
	\subparagraph{A guessing game.}
	Let $(X_{j,\sigma })_{j \in  \mathbb{N}, \sigma  \in  [3]}$ be a family of i.i.d.\ random variables with uniformly random values in $[n]$. Consider a game in which the player does not know the random variables. She may issue queries of the form $(i,j) \in  [n] \times  \mathbb{N}$ and is told whether or not $i \in  \{ X_{j,\sigma } \mid \sigma  \in  [3]\}$. If she gets a positive answer, we say she has \emph{solved} group $j$ and she is told the values $X_{j,1}$,$X_{j,2}$ and $X_{j,3}$.
	\begin{lemma}
		\label{lem:abstract-game}
		Assume a player is given $6n$ queries. She solves more than $r$ groups with probability $\mathcal{O}(2^{-r})$.
	\end{lemma}
	\begin{proof}
		Consider a query $(i,j) \in  [n] \times  \mathbb{N}$ that is the $k$th query of the form $(\cdot ,j)$ where $j$ is an unsolved group. We distinguish two cases.
		\begin{description}
			\item [Case 1: $\bm{k \geq  n/2}$.]
			We generously grant group $j$ as solved. This affects at most $6n/(n/2) = 12$ groups.
			\item [Case 2: $\bm{k \leq  n/2}$.]
			Up to $k-1$ values are already excluded for the variables $X_{j,1}$, $X_{j,2}$ and $X_{j,3}$. All other values are equally likely. The probability of solving group $j$ with this query is at most $1-(1-\frac{1}{n-k+1})^3 < 1-(1-\frac{2}{n})^3 \leq  \frac{6}{n}$.
			
			Let $Y$ be the number of successes due to Case 2. 
			A simple coupling yields a random variable $Z$ with $Y \leq  Z$ and $Z \sim \Bin(6n,\frac{6}{n})$. We can think of $Z$ as the sum of the outcomes of $6n$ Bernoulli trials. In order for $Z$ to attain a value of at least $\ell  = r-12$ there must exist a set $T \subseteq  [6n]$ of size $\ell $ such that the corresponding Bernoulli trials are all successes. A union bound over all choices of $T$ yields:
			\begin{align*}
				\Pr[Y \geq  \ell ] \leq  \Pr[Z \geq  \ell ] \leq  \binom{6n}{\ell }\bpfrac{6}{n}^{\ell }
				\leq  \bpfrac{6ne}{\ell }^{\ell } \bpfrac{6}{n}^{\ell } \leq  \bpfrac{36e}{\ell }^{\ell } = \bpfrac{36e}{r-12}^{r-12} \leq  2^{-r+12}.
			\end{align*}
			where the last step assumes $r \geq  72e+12$.
		\end{description}
		Overall, if $r \geq  72e+12$ then the number of successes is bounded by $r$, except with probability $\mathcal{O}(2^{-r})$.
	\end{proof}
	
	\subparagraph{How this relates to the foreign keys.}
	
	Let $F \subseteq  \mathcal{U} \setminus S$ be the set of all foreign keys that are added to $\Sdec$ at least once during the execution of decode.
	\begin{lemma}
		\label{lem:few-foreign-keys}
		For any $r \in  \mathbb{N}$ we have $\Pr[|F| > r] = \mathcal{O}(2^{-r})$.
	\end{lemma}
	
	\begin{proof}
		Before decoding begins, the family $(h_\sigma (y))_{y \in  \mathcal{U} \setminus S, \sigma  \in  [3]}$ of hashes of foreign keys are independent and uniformly random in $[n]$ just as is required in the guessing game above. We interpret \op{decode} as a player. Whenever \op{decode} executes the \op{looksPure} operation for a bucket $i$ containing $A[i] = y \in  \mathcal{U} \setminus S$, then we interpret this as the player issuing the query $(y,i)$. He learns whether $i \in  \{h_1(y),h_2(y),h_3(y)\}$ and, on a positive answer, he also learns the values $h_1(y),h_2(y),h_3(y)$. This is sufficient to continue the execution of \op{decode}, in particular decode amounts to a valid player respecting the access limitation to the relevant random variables.
		
		The crucial observation is that any $y \in  F$ is solved by \op{decode}. Indeed, if $y \in  F$ then $y$ was added to $\Sdec$ at some point during the execution of \op{decode} and this addition was immediately preceeded by a corresponding successful \op{looksPure} check.
		
		The number $q$ of queries issued by \op{decode} is by definition the number of calls to \op{looksPure}. Every time a bucket $i$ is considered to be added to $Q$ or $\QN$ there is one such call, and, if $i$ is added, there can be another call when $i$ is removed from $Q$. Due to the time limit we have $q \leq  2(n+t_{\max}) \leq  6n$. Therefore \cref{lem:abstract-game} implies our claim.
	\end{proof}
	
	\subsection{Bounding Centres of Internal Anomalies}
	\label{sec:internal-anomalies}
	
	It is worthwhile to clarify the notion of anomalies by giving a precise definition of the set $\mathcal{A}$ of all anomalies.
	\[ \mathcal{A} = \Big\{ \varnothing  \neq  A \subseteq  \mathcal{U} \Big\vert \bigoplus_{x \in  A} x = 0, \bigcap_{x \in  A} h(x) \neq  \varnothing \Big\}\]
	If $A = \{x_1,\dots ,x_\ell \}$ is an anomaly then the presence of any $\ell -1$ keys from $A$ can be mistaken as the presence of the missing $\ell $th key (because of $x_1\oplus \dots \oplus x_{\ell -1} = x_\ell $). A value $i \in  \bigcap_{x \in  A} h(x)$ is a centre of the anomaly $A$ (in principle a single anomaly could have multiple centres). The set $\mathcal{A}$ makes no mention of $S$ and contains many anomalies that are unlikely to become relevant during decoding. 
	Immediately threatening is, however, the set of \emph{native anomalies} $\mathcal{A}_{\mathrm{nat}}$ define as
	\[ \mathcal{A}_{\mathrm{nat}} = \{ A \in  \mathcal{A} \mid |A \setminus S| \leq  1\}.\]
	For a native anomaly $A$, all but at most $1$ key from $A$ are present in the beginning, enough to cause trouble right away. In what follows, a subset of the native anomalies, namely the \textit{internal} anomalies $\mathcal{A}_{\mathrm{int}}$ and the set $N$ of its centres play a role. They are defined as
	\[ \mathcal{A}_{\mathrm{int}} = \{A \in  \mathcal{A} \mid A \subseteq  S\} \text { and } N = \bigcup_{A \in  \mathcal{A}_{\mathrm{int}}} \bigcap_{x \in  A} h(x). \]
	
	\begin{lemma}
		\label{lem:few-internal-anomalies}
		For any $r \in  \mathbb{N}$ we have $\Pr[|N| > r] = \mathcal{O}(2^{-r})$.
	\end{lemma}
	\begin{proof}
		The authors of \cite{HPW22} show in their Lemma 2.2 that the set $N_{\mathrm{nat}} = \bigcup_{A \in  \mathcal{A}_{\mathrm{nat}}} \bigcap_{x \in  A} h(x)$ of centres of native anomalies satisfies
		$\Pr[N' \geq  r] \leq  \bpfrac{3e^4}{r}^r$.
		For $r \geq  6e^4$ this is at most $2^{-r}$. Since $N_{\mathrm{nat}} \supseteq  N$ the claim follows.
	\end{proof}
	
	
\subsection{Early Rounds of the Peeling Process}
\label{sec:early-peeling}

We will use a result by Molloy who analysed cores of random hypergraphs \cite{molloy2005cores}. The hypergraph $H$ arising in our setting has vertex set is $[n]$ and the $m$ hyperedges $\{\{h_\sigma (x) \mid \sigma  \in  [k] \} \mid x \in  S\}$. By the \emph{parallel peeling} process we mean an algorithm on a hypergraph. In each round, all vertices of degree $0$ or $1$ are determined (simultaneously) and then all of these vertices are removed, including all incident hyperedges. When no vertex of degree $0$ or $1$ remains, the core of the hypergraph is reached, which may be empty.
Note that decode implements this peeling process, except that anomalies may occur.

\subparagraph{Early Peeling Without Anomalies.}
A lemma by Molloy is the following.

\begin{lemma}[{\cite[Lemma 3]{molloy2005cores}}]
	\label{lem:molloy}
	For any $\varepsilon ,\delta > 0$, there exists $R \in  \mathbb{N}$ such that after peeling a $k$-uniform hypergraph with 
	$\frac{n}{m} > c_3+\varepsilon $ 
	for $R$ rounds then at most $\delta n$ vertices remain in expectation.\footnote{Molloy's Lemma is slightly stronger in that it claims that at most $\delta n$ vertices remain with probability $1-o(1)$. Since we need a much stronger concentration bound, we will redo a corresponding step ourselves.}
\end{lemma}

In the following we assume that constants $\varepsilon ,\delta  > 0$ are given and $R \in  \mathbb{N}$ is the corresponding constant from \cref{lem:molloy}.

By the $R$-neighbourhood of a vertex $v$ or hyperedge $e$ we mean the subhypergraph of $H$ induced by all vertices that can be reached from $v$ or $e$ by traversing at most $R$ hyperedges.

It is a well-known fact (and a consequence of the Poisson limit theorem) that the degree distribution of $H$ converges to a Poisson distribution with parameter $3/c_3$ for $n \rightarrow  \infty $. More generally, the distribution of the $R$-neighbourhoods of the vertices in $H$ converges (see e.g.\ \cite{K:Poisson:2006,L:Belief_Propagation:2013}). This implies that we can choose a constant $D = D(\varepsilon ,\delta ,R)$ large enough such the probability that a given vertex has a vertex of degree at least $D$ in its $R$-neighbourhood is at most $\delta $. 

In this context call a vertex $v$ \emph{good} if
\begin{enumerate}[(a)]
	\item  $v$ is removed by the $R$-round peeling process
	\item  the $R$-neighbourhood of $v$ does not contain a vertex of degree at least $D$.
\end{enumerate}
Since the first condition is met by at least $(1-\delta )n$ vertices in expectation by \cref{lem:molloy} and the second condition by at least $(1-\delta )n$ vertices in expectation by choice of $D$, the set $G$ of good vertices has expectation at least $(1-2\delta )n$.

We will now use the bounded difference inequality by McDiarmid \cite{McDiarmid:1989}.
\begin{lemma}[McDiarmid's inequality \cite{McDiarmid:1989}]
	\label{lem:McDiarmid}
	Let $\mathcal{X}$ be some set, $d \in  \mathbb{N}$ and $f: \mathcal{X}^n \rightarrow  \mathbb{R}$ a function such that changing any one argument of $f$ can affect the function value by at most $d$. Let $X_1,\dots ,X_n$ be i.i.d.\ random variables with values in $\mathcal{X}$. Then
	\[
	\Pr[|f(X_1,\dots ,X_n) - \mathbb{E}[f(X_1,\dots ,X_n)]| > \varepsilon ] \leq  \exp(-\varepsilon ^2/(nd^2)).
	\]
\end{lemma}
This gives us a concentration bound on the number of good vertices.
\begin{lemma}
	\label{lem:many-good-vertices}
	In the above context there exists a constant $\gamma  = \gamma (\varepsilon ,\delta ,R,D) > 0$ such that
	\[ \Pr[|G| < (1-3\delta )n] \leq  \exp(-\gamma n).\]
\end{lemma}

\begin{proof}
	Note that whether or not a vertex is good is a \emph{local} property in the sense that it only depends on its $R$-neighbourhood. This gives us a bounded difference property of $|G|$ as follows.
	
	The set $G$ (and therefore $|G|$) is a function of the $3m$ hash values defining the $m$ hyperedges. Assume we change a single hash value from $v$ to $v' \neq  v$. A previously non-good vertex might become good only if it was in the vicinity of $v$ and a previously good vertex might become non-good only if it ends up in the vicinity of $v'$. More precisely, a vertex $w$ might be affected if prior to the change, $v$ or $v'$ (or both) were reachable from $w$ via a path of length at most $R$ that only traverses vertices of degree at most $D$ (vertices reachable through vertices of higher degree are non-good anyway). To bound the number of such paths, note that in every step we can first choose among $\leq  D$ incident hyperedges and then among $3-1$ endpoints of the chosen hyperedge, or we can choose to end the path at the current vertex.
	In particular, changing a single incidence affects the number of good vertices by at most $2(2D+1)^R$. We can therefore apply McDiarmid's inequality (\cref{lem:McDiarmid}) to conclude that 
	\[
	\Pr[n-|G| > 3\delta n] \leq  \Pr\big[(n-|G|) - \mathbb{E}[n-|G|] > \delta n\big] \leq  \exp(-(\delta n)^2/(n\cdot (4(2D-1)^{2R})) = \exp(-\gamma n)
	\]
	when defining $\gamma  = \delta ^2/(4(2D+1)^{2R})$.
\end{proof}

\subparagraph{Early Peeling With Anomalies.} We now get back to the analysis of decode. Let $\Sdec^R$ be the state of $\Sdec$ after $R$ rounds of the while-loop are executed (ignoring the time limit). Let $S^R := \Sdec^R \triangle  S$ be the set of elements stored in the IBLT at that time and $I^R := \bigcup_{x \in  S^R}h(x)$ the set of buckets touched by $S^R$.

\begin{lemma}
	\label{lem:no-large-core}
	For any $\varepsilon  > 0$ and $\delta  > 0$ there exist constants $R$ and $\beta $ such that
	\[\Pr[|I^R| > 7\delta n] \leq  \mathcal{O}(2^{-\beta n}).\]
\end{lemma}

\begin{proof}
	Given $\varepsilon $ and $\delta $, let $R=R(\varepsilon ,\delta )$, $D=D(\varepsilon ,\delta ,R)$ and $\gamma  = \gamma (\varepsilon ,\delta ,R,D)$, be the corresponding constants from the discussion of the discussion above. Recall from \cref{lem:few-foreign-keys,lem:few-internal-anomalies,lem:many-good-vertices} that $F$ is the set of foreign keys added to $S$ at least once, $N$ is the set of centres of internal anomalies and $G$ is the set of good vertices. We know:

	\begin{tabular}{rlc}
		$|F| \leq  \frac{\delta n}{(2D+1)^{R}}$
		& except with probability $\mathcal{O}\Big(2^{-\frac{\delta n}{(2D+1)^R}}\Big)$
		& by \cref{lem:few-foreign-keys}\\
		$|N| \leq  \frac{\delta n}{(2D+1)^{R}}$
		& except with probability $\mathcal{O}\Big(2^{-\frac{\delta n}{(2D+1)^R}}\Big)$
		& by \cref{lem:few-internal-anomalies}\\
		$|G| \geq  (1-3\delta )n$
		& except with probability $\exp(-\gamma n)$
		& by \cref{lem:many-good-vertices}\\
	\end{tabular}

	\noindent
	The sum of the error probabilities is $\mathcal{O}(2^{-\beta n})$ where $\beta  = \delta /(2D+1)^R+\gamma \log_2(e)$. Thus we may assume that all three stated events occur.
	
	We now upper bound $|I^R|$. We pessimistically assume that each $i \in  [n]$ that is not good is in $I^R$, which accounts for at most $3\delta n$ vertices. Any good $i \in  [n]$ would be removed by the proper peeling process, so if $i \in  I^R$ then an anomaly must have interfered. Therefore, the $R$-neighbourhood of $i$ must include the centre of an internal anomaly or a vertex incident to a key from $F$. In other words, the $R$-neighbourhood of $i$ must include a vertex from $|N \cup  h(F)|$. At most $|N \cup  h(F)|\cdot (2D+1)^R$ good vertices can be affected in this way, namely those reachable from a vertex in $N \cup  h(F)$ via a path of length at most $R$ traversing only vertices of degree at most $D$. Taken together the size of $I_R$ is bounded by
	\[ |I_R| \leq  3\delta n + |N \cup  h(F)| \cdot  (2D+1)^R \leq  3\delta n + \frac{4\delta }{(2D+1)^R} n \cdot  (2D+1)^R \leq  7\delta n.\qedhere\]
\end{proof}


\subsection{Late Rounds of the Peeling Process}
\label{sec:late-peeling}

\begin{lemma}
	\label{lem:no-small-core}
	There exists a constant $\delta $ such that for any $1 \leq  r \leq  n$ there is an event $E$ with probability $\mathcal{O}(2^{-r})$ such that, when $E$ occurs, then
	\[
	\forall I \subseteq  [n] \text{ with } 5r \leq  |I| \leq  \delta n: |S[I]| \leq  \tfrac{3}{5}|I|.
	\]
\end{lemma}

\begin{proof}
	We use a union bound over all possible sizes $i$ of $I$ and all possible sets of size $i$. Let $p_i$ be the probability that a violating set of size $i$ exists and $j := \frac 35i$. Multiplying the number $\binom{n}{i}$ of choices for $I$, the number $\binom{m}{j}$ of choices for $T \subseteq  S$ of size $j$ and the probability for $T \subseteq  S[I]$ (i.e.\ the probability that all $x \in  T$ satisfy $h(x) \subseteq  I$) gives
	\begin{align*}
		p_i &\leq  \binom{n}{i}\binom{m}{j} \bpfrac{i}{n}^{3j}
		\leq  \bpfrac{ne}{i}^i \bpfrac{ne}{j}^j \bpfrac{i}{n}^{3j}
		= e^{i+j} \bpfrac{i}{j}^j \bpfrac{i}{n}^{-i+2j}
		\leq  e^{3i} \bpfrac{i}{n}^{i/5} = \bpfrac{e^{15}i}{n}^{i/5}.
	\end{align*}
	Setting $\delta  = 1/(2e^{15})$ we can bound the overall failure probability $p$ as
	\begin{align*}
		p &\leq  \sum_{i = 5r}^{\delta n} p_i
		\leq  \sum_{i = 5r}^{\delta n} \bpfrac{e^{15}i}{n}^{i/5}
		\leq  \sum_{i = 5r}^{\delta n} \bpfrac{e^{15}\cdot \delta n}{n}^{i/5}
		\leq  \sum_{i = 5r}^{\delta n} (\tfrac{1}{2})^{i/5}\\
		&= \sum_{i = 5r}^{\delta n} \Big((\tfrac{1}{2})^{1/5}\Big)^{i}
		\leq  2^{-r}\cdot \sum_{i = 0}^{\infty } \Big((\tfrac{1}{2})^{1/5}\Big)^{i} = 2^{-r}\cdot \theta (1) = \mathcal{O}(2^{-r}).\qedhere
	\end{align*}
\end{proof}	

\subsection{Proof of Theorem \ref{thm:main-inner}}
\label{sec:main-proof}

We now assemble the previous Lemmas into a proof of \cref{thm:main-inner}. Note that we only claimed an error probability of $2^{-\Omega (r)}$, not $\mathcal{O}(2^{-r})$. By a change of parameters it therefore suffices that we show $\Pr[|\Sdec \triangle  S| > 185r] = \mathcal{O}(2^{-r})$.

The parameter $\varepsilon  > 0$ and $r = o(n)$ are given. We use the absolute constant $\delta  > 0$ from \cref{lem:no-small-core}. We apply \cref{lem:no-large-core} for $\varepsilon $ and $\delta ' = \delta /7$ which gives us $R \in  \mathbb{N}$ and $\beta  > 0$. For convenience, let us summarise the statements used in the following. Since each holds with probability $1-\mathcal{O}(2^{-r})$ we need only show that they imply $|\Sdec \triangle  S| \leq  185r$.
\begin{align}
	&|I^R| < \delta n
	&& \text{by \cref{lem:no-large-core}}\label{eq:early-peel}\\
	&\forall I \subseteq  [n]\text{ with }5r \leq  |I| \leq  \delta n: |S[I]| \leq  \tfrac{3}{5}|I|
	&& \text{by \cref{lem:no-small-core}}\label{eq:no-small-core-eq}\\
	&|N| \leq  r\text{ and }|F| \leq  r
	&& \text{by \cref{lem:few-internal-anomalies,lem:few-foreign-keys}}
\end{align}
Recall that \emph{anomalous decoding steps} are those that add a key rather than removing a key. 
\begin{claim}
	\label{claim:anomalous-steps}
    There are at most $|N|+3|F| \leq  4r$ anomalous decoding steps per round.
\end{claim}
\begin{proof}
    There are at most $3|F|$ buckets in which a key from $F$ might be placed. We may conservatively assume that there is an anomalous decoding step in each round for each of these buckets. Anomalous decoding steps where no foreign key is involved can arise, by definition, only in centres of internal anomalies, hence there can be at most $|N|$ per round. 
\end{proof}

We are most interested in the first $\hat{\rho }$ rounds of \op{decode} where $\hat{\rho } := R + \log_{\frac{20}{19}}(\frac{n}{300r}) = \theta (\log \frac{n}{r})$. We now make sure that we are referring to well-defined rounds. Execution might end due to $Q = \varnothing $, or due to the time limit. If execution ends prior to round $\hat{\rho }$ due to $Q = \varnothing $, just pretend the loop is executed for a suitable number of additional rounds. Note that an additional round does not perform any steps when $Q=\varnothing $. Let us now deal with the time limit.

\begin{claim}
	\label{claim:time-limit}
	The limit $t_{\max} = 2n$ on steps does not take effect within the first $\hat{\rho }$ rounds.
\end{claim}
\begin{proof}
	We use a potential argument. There are $n$ keys in the beginning. Each non-anomalous decoding step removes a key and each anomalous decoding step adds a key. Hence, if $a$ denotes the number of anomalous decoding steps in the first $\hat{\rho }$ rounds then the total number of decoding steps is at most $n+2a$. The previous claim gives $a \leq  4r\hat{\rho }$. We can therefore bound the number of decoding steps in the first $\hat{\rho }$ rounds by
	\[
		n + 2a \leq  n + 8r\hat{\rho } = n(1+\theta (\tfrac{r}{n}\log \tfrac{n}{r})) = n(1+o(1)) \leq  2n.
	\]
	The last step uses $r = o(n)$ and $\lim_{x \rightarrow  0} x \log \frac{1}{x} = 0$. 
\end{proof}

We can now safely refer to each round $\rho  \in  \{1,\dots ,\hat{\rho }\}$. Let $S^\rho $ be the set of keys stored in the IBLT after $\rho $ rounds, $I^\rho  = h(S^\rho )$ the set of buckets touched by these keys and $S^\rho _0 = S^\rho  \setminus F$.
\begin{claim}
	\label{claim:geometric-decrease}
	If $300r \leq  |I^\rho | \leq  \delta n$ then $|I^{\rho +1}| \leq  \frac{19}{20}|I^\rho |$.
\end{claim}

\begin{proof}
    Since $S^\rho _0 \subseteq  S[I^\rho ]$ we know $|S^\rho _0| \leq  \frac{3}{5}|I^\rho |$ from \cref{eq:no-small-core-eq}. The number of incidences due to the keys from $S^\rho _0$ is $3|S^\rho _0| \leq  \frac{9}{5}|I^\rho | = \frac{9}{10}\cdot 2|I^\rho |$. Thus at most a $\frac{9}{10}$-fraction of the buckets in $I^\rho $ have two or more incidences to keys from $S^\rho _0$. This leaves at least $\frac{1}{10}|I^\rho |$ buckets with at most one incidence to $S^\rho _0$. If $i$ is such a bucket, we would normally expect $i \notin  I^{\rho +1}$. 

    There are only two exceptions. A foreign key might be stored in $i$ at some point during round $\rho +1$ or a key from $S$ is added back into bucket $i$ at some point during round $\rho +1$. Each foreign key and each anomalous decoding step affects at most $3$ buckets. Hence, by an earlier claim, at most $3(|F|+4r) \leq  15r$ buckets are exceptional in this sense. Together we get
\[
	|I^{\rho +1}| \leq  |I^{\rho }| - \tfrac{1}{10}|I^\rho | + 15r \leq  \tfrac{9}{10}|I^\rho | + 15r \leq  \tfrac{19}{20}|I^\rho |.
\]
	where the last step uses $|I^\rho | \geq  300r$.
\end{proof}
\begin{claim}
	\label{claim:stability}
	If $|I^\rho | \leq  300r$ then $|I^{\rho +1}| \leq  300r$.
\end{claim}
\begin{proof}
    Assume w.l.o.g.\ that $|I^\rho | \geq  5r$ and modify the last step of the previous claim:
    \[\tfrac{9}{10}|I^\rho | + 15r \leq  \tfrac{9}{10}\cdot 300r+15r = 285r < 300r.\qedhere\]
\end{proof}
Let now $\rho _{\max}$ be the last round that is fully executed.
\begin{claim}
	\label{claim:last-round-small}
    $|I^{\rho _{\max}}| \leq  300r$.
\end{claim}
\begin{proof}
	We have $|I^R| < \delta n$ by \cref{eq:early-peel}. In subsequent rounds $\rho  \in  \{R+1,R+2,\dots \}$, the size of $I^\rho $ decreases by at least a factor of $\frac{19}{20}$ until being at most $300r$ by \cref{claim:geometric-decrease}. By choice of $\hat{\rho }$ we have $I^{\rho } \leq  300r$ for some $\rho  \leq  \hat{\rho }$. By \cref{claim:stability} $|I^{\rho }|$ will subsequently not rise above $300r$ for larger values of $\rho $.
\end{proof}
By the last claim and \cref{eq:no-small-core-eq} we have
\[|S^{\rho _{\max}}_0| \leq  S[|I^{\rho _{\max}}|] \leq  \tfrac{3}{5}\cdot (\max(5r,|I^{\rho _{\max}}|)) \leq  \tfrac{3}{5}\cdot 300r = 180r.\]
The set $S^{\rho _{\max}}$ of all keys after round $\rho _{\max}$ might additionally include up to $|F| \leq  r$ foreign keys. Since an additional round might be started but cut short due to the time limit, another $4r$ keys might be added due to anomalous decoding steps. Overall the set $S \triangle  \Sdec$ of keys returned in the end has size at most $180r+r+4r = 185r$. This concludes the proof.

\section{Set reconciliation with error-correcting codes}
\label{sec:err_correct}
The set reconciliation problem can be modeled by encoding a set as a binary string of length $U$ which is then sent through a channel inflicting up to $D$ errors (bitflips).
Recovering the errors is then equivalent to reporting the keys inserted to or deleted from the set. This reduction immediately implies that the problem can be solved by an appropriate application of error-correcting codes, in particular a linear code such as BCH, Reed-Solomon or expander code (see  \cite{karpovsky2003data}). 
This reduction, in turn, highlights the relationship of set reconciliation to the problem of \textit{sparse recovery}: reconstructing a sparse vector from a set of linear measurements (see e.g. \cite{cheraghchi2019simple}). 

We will use a BCH code over the field $GF(2^w)$ with  $U=2^w-1$. It is known (see e.g. \cite{MacWilliamsSloane}) that if up to $D$ errors have to be corrected, the binary parity-check matrix $\mathcal{H}$ for BCH has dimensions $Dw\times U$. Let each set $S$ be conceptually represented by a binary vector of size $U$. Denote $\mathcal{H}(S)$ the image of $S$ by $\mathcal{H}$. Then for two sets $S,T$, $\mathcal{H}(S) \oplus \mathcal{H}(T)=\mathcal{H}(S\triangle T)$. That is, if $|S\triangle T|\leq D$, then $\mathcal{H}(S\triangle T)$, called \textit{syndrome}, encodes the differences between $S$ and $T$ and can be reconstructed by a syndrome decoding algorithm. This implies that $\mathcal{H}(S)$ can be defined as a sketch of $S$ of $Dw$ bits so that the differences between two sets can be recovered from the XOR of their sketches, provided that there are at most $D$ of them. 

From the algorithmic viewpoint, inserting a key amounts to XORing the sketch with the corresponding column of $\mathcal{H}$ composed of $D$ blocks of $w$ bits (elements of $GF(2^w)$). Computing each block 
amounts to 
$D$ multiplications in $GF(2^w)$\footnote{We refer to \cite{DBLP:journals/siamcomp/DodisORS08} for details on how keys of $\mathcal{U}$ are represented as elements of $GF(2^w)$}. \gkk{A multiplication in $GF(2^w)$ can be done with $\mathcal{O}(w\log w)$ bit operations \cite{DBLP:journals/jacm/HarveyH22}, but only with $\mathcal{O}(w)$ operations in the RAM model (see Appendix). 
Thus, an insertion of a key takes time $\mathcal{O}(Dw)=\mathcal{O}(D\log U)$. }

Decoding is a more complex operation. The efficient syndrome decoding algorithm of \cite{DBLP:journals/siamcomp/DodisORS08} requires $\mathcal{O}(D^2 w)$ multiplications in $GF(2^w)$ resulting in \gkk{$\mathcal{O}(D^2 w^2) = \mathcal{O}(D^2 \log^2 U)$} time for decoding. Note that if the number $d<D$ of errors is known, then by properties of syndromes, it is sufficient to decode the first $dw$ bits of the sketch, and the decoding complexity becomes \gkk{$\mathcal{O}(d^2 \log^2 U)$}. 

Denote $\mathtt{toggle\_BCH}(x,C)$ and  $\mathtt{decode\_BCH}(C)$ the operation of inserting/deleting a key $x$ to/from a BCH sketch $C$ and decoding a BCH sketch $C$, respectively.  Complexity of BCH sketches is summarized in the following theorem. 
\begin{theorem}
	\label{thm:BCH_sketch}
	Consider a set $S$ containing no more than $D$ keys from $U$. Let $S$ be stored in a BCH sketch $C$ of $D\log (U+1)$ bits of space. 
	Then
	\begin{itemize}
		\item $\mathtt{toggle\_BCH}(x,C)$ takes \gkk{$\mathcal{O}(D\log U)$} time,
		\item $\mathtt{decode\_BCH}(C)$ decodes keys of $S$ with no error in \gkk{$\mathcal{O}(D^2 \log^2 U)$} time.
	\end{itemize}
\end{theorem}
Reconciliation of two sets is done by simply XORing their sketches $C_1$ and $C_2$, which can be trivially sped up by packing bits in words. We denote this operation by $\mathtt{merge\_BCH}(C_1,C_2)$. A software implementation of a BCH sketch for set reconciliation is available \cite{minisketch}. 	

	\section{IBLT with stash}
	\label{sec:IBLT_with_stash}
	
	Our goal is to enhance the IBLT in order to solve set reconciliation with a much stronger success guarantee than that provided by IBLT alone (\cite{HPW22}, Table~\ref{table:1}) while keeping small space and fast decoding time provided by the IBLT approach. Given an upper bound $D$ on the size of set differences, our sketch $\mathcal{S}$ consists of three components \gkk{depending on} a parameter $r$: 
	\begin{itemize}
		\item an IBLT $\mathcal{S}.A$ of size  $n=(c_3+\varepsilon)D$ 
		with three hash functions, as defined in Section~\ref{sec:def_IBLT},
		\item a control checksum $\mathcal{S}.H$ of $r$ bits, 
		\item a \textit{stash} data structure $\mathcal{S}.C$ defined as 
		a BCH syndrome of $r\log U$ bits (see Section~\ref{sec:err_correct}). 
	\end{itemize}
	The control checksum $\mathcal{S}.H$ is used to check if the decoding of the main IBLT $\mathcal{S}.A$ succeeded. 
	Before the decoding, $\mathcal{S}.H=\bigoplus_{x\in S}h(x)$ where $h:\mathcal{U}\rightarrow 2^r$ is a random hash function and $S$ is the currently stored key set. 

	The stash $\mathcal{S}.C$ is configured so that it can decode up to $r$ keys with no error, as described in Section~\ref{sec:err_correct}. 
	%
	The entire sketch takes $(c_3+\varepsilon)D\log U+r(1+\log U)=(c_3+\varepsilon)D\log U+o(D)$ bits. 
	
		\begin{figure}
		\begin{minipage}{0.47\textwidth}
			\begin{algorithm}[H]
				\algo{\texttt{insert}($x,\mathcal{S}$)}{
					$\mathtt{toggle}(x)$ in $\mathcal{S}.A$ \;
					$\mathcal{S}.H\gets \mathcal{S}.H\oplus h(x)$\;
					$\mathtt{toggle\_BCH}(x,\mathcal{S}.C)$ 
				}
			\end{algorithm}
		\begin{algorithm}[H]
			\algo{\texttt{diff}($\mathcal{S}_1,\mathcal{S}_2$)}{
				$\hat{\mathcal{S}}.A\gets \mathtt{merge}(\mathcal{S}_1.A,\mathcal{S}_2.A)$ \;
				$\hat{\mathcal{S}}.H\gets \mathcal{S}_1.H\oplus \mathcal{S}_2.H$\;
				$\hat{\mathcal{S}}.C\gets \mathtt{merge\_BCH}(\mathcal{S}_1.C, \mathcal{S}_2.C)$ \;
				\Return{$\hat{\mathcal{S}}$}
				}
		\end{algorithm}
		\end{minipage}
	   \begin{minipage}{0.53\textwidth}
	   	\begin{algorithm}[H]
	   		\algo{\texttt{report}($\mathcal{S}$)}{
	   			$\Sdec\gets \mathtt{decode}(\mathcal{S}.A)$\;
	\For{$x \in  \Sdec$}{
		$\mathcal{S}.H\gets \mathcal{S}.H\oplus h(x)$ 
	}
	\If{$\mathcal{S}.H\neq 0$}{
	   			\For{$x \in  \Sdec$}{
	   				$\mathtt{toggle\_BCH}(x,\mathcal{S}.C)$ \tcp{delete $x$ from $\mathcal{S}.C$} 
	   			}
	   			$\Sdec'\gets \mathtt{decode\_BCH}(\mathcal{S}.C)$
	   		}
		\Return{$\Sdec\triangle \Sdec'$}
	   		}
	   	\end{algorithm}
	   \end{minipage}
		\caption{Sketch operations. $\mathtt{insert}(x,\mathcal{S})$ inserts key $x$ to sketch $\mathcal{S}$, $\mathtt{diff}(\mathcal{S}_1,\mathcal{S}_2)$} computes the difference of sketches $\mathcal{S}_1$ and $\mathcal{S}_2$, $\mathtt{report}(\mathcal{S})$ reports keys stored in $\mathcal{S}$
		\label{fig:sketch}
	\end{figure}
	
		Figure~\ref{fig:sketch} shows Algorithms for inserting a key to a sketch, computing a sketch difference, and reporting keys. 
		
	
	A run of \texttt{report}($\mathcal{S}$) (Figure~\ref{fig:sketch}) can follow \gkk{two} scenarios. If the number of keys in the sketched set is no more than $D$, the main IBLT recovers those keys with probability $1-\tilde{O}(1/D)$ 
	(Theorem~\ref{thm:main-inner}). With probability $\tilde{O}(1/D)$, however, the recovery may fail by ``getting stuck'', by resulting in a ``falsely empty'' IBLT (see \cite{HPW22}), or by being aborted after $2n$ steps. Checking the checksum ensures that an incorrect decoding will be recognized except with probability $\mathcal{O}(1/2^r)$. By Theorem~\ref{thm:main-inner}, with probability $2^{-\Omega(r)}$, the output $\Sdec$ differs from $S$ by no more than $r$ missing or foreign keys. 
	The algorithm resorts then to the stash in order to correct the output, i.e.\ to recover the missing/superfluous keys, under assumption that there are at most $r$ of them. In this case, the stash reports those keys with no error. 
		
	The execution time of \texttt{report} depends on whether the correction step is triggered or not. If the decoding of the main IBLT is deemed to be successful by the checksum, the time taken will be $\mathcal{O}(D(1+\frac{r}{w}))$, where $w$ is the wordsize. \gkk{Here we assume that updating the checksum is done by bit packing and takes $\mathcal{O}(\frac{r}{w})$ time.} Otherwise the algorithm will spend additional \gkk{$\mathcal{O}(r\log U)$} time on performing $\mathtt{toggle\_BCH}$ on each of $\mathcal{O}(D)$ key in $\Sdec$, followed by the decoding taking additional \gkk{$\mathcal{O}(r^2\log^2 U)$} time. Since the stash is activated with probability $\tilde{O}(1/D)$, the expected time is \gkk{$\mathcal{O}(D(1+\frac{r}{w})+r\log U+\frac{r^2 \log^2 U}{D})$}. 
	Assuming 
		\gkk{$r=\min\left(\mathcal{O}(D/\log^2 U),\mathcal{O}(\log U)\right)$}, the time of $\mathtt{report}$ is $\mathcal{O}(D)$.
	
	We summarize the properties of \texttt{report} in the following theorem. 
	\begin{theorem}
		Let $\mathcal{S}$ be a sketch built for a key set $S$ and \gkk{$r=\min\left(\mathcal{O}(D/\log^2 U),\mathcal{O}(\log U)\right)$}. Then
		\begin{itemize}
			\item $\mathcal{S}$ takes $(c_3+\varepsilon)D\log U + r(1+\log U)=(c_3+\varepsilon)D\log U+o(D)$ bits,
			\item inserting an element to $\mathcal{S}$ takes time \gkk{$\mathcal{O}(r\log U)$},
			\item when $|S|\leq D$, $\mathtt{report}(\mathcal{S})$ correctly recovers $S$ with probability $1-2^{-\Omega(r)}$, 
			\item $\mathtt{report}(\mathcal{S})$ takes 
			$\mathcal{O}(D)$ expected time. 
		\end{itemize}
		\label{thm:main}
	\end{theorem}
	
	Note that if $r=\omega(\log D)$, the failure probability for $\mathtt{report}$ to recover a set of up to $D$ keys is $o(1/D)$, as opposed to $\tilde{O}(1/D)$ for recovery without stash \cite{HPW22}. 
	
	For two sets $S$ and $T$ represented by their respective sketches $\mathcal{S}_S$ and $\mathcal{S}_T$, \texttt{diff}($\mathcal{S}_S,\mathcal{S}_T$) is the sketch of symmetric difference $S\Delta T$.  Therefore, Theorem~\ref{thm:main} applies directly to the set reconciliation problem. In particular, if $|S\Delta T|\leq D$, then $S\Delta T$ can be recovered with guarantees stated in Theorem~\ref{thm:main}.

	\section{Distinguishing \texorpdfstring{$S\setminus T$ and $T\setminus S$}{S \textbackslash T and T \textbackslash S}}
	\label{sec:distinguishing}
	
	
	The algorithm we presented computes symmetric difference $S\Delta T$ but is not capable of distinguishing elements of $S$ and $T$ in the output. However, the latter is desirable for many applications. Here we outline how our sketch \gk{and set reconciliation protocol (Figure~\ref{fig:sketch})} can be modified in order to make this possible. \gk{The general idea is to assign a ``sign'' to keys depending on whether they come from $S$ or $T$ and to keep track of it in a consistent manner. Formal proofs are left to the full version of the paper.}
	
	\subparagraph{Modified IBLT} 
	\gk{
	We modify our IBLT (Section~\ref{sec:def_IBLT}) as follows. IBLT entries will now be \textit{ternary} strings of $\{0,1,2\}^\nu$ seen as elements of group $\mathbb{Z}_3^\nu$ (additive group of $GF(3^\nu))$. The first \textit{trit} encodes a \textit{sign} and the other $\nu-1$ trits encode a keysum. We use an appropriate encoding of $\mathcal{U}$ into strings $\{0,1,2\}^{\nu-1}$ and denote by $\widetilde{x}$ the encoding of $x$. 
	
	Toggling a key $x$ is replaced by two operations: insertion and deletion. Inserting $x$ into an IBLT is done by adding (in $\mathbb{Z}_3^\nu$) $1\widetilde{x}$ to each of the three entries at $h(x)$, and deleting $x$ is done by subtracting $1\widetilde{x}$ from each of the three entries at $h(x)$ or, equivalently, adding $2(-\widetilde{x})\equiv (-1)(-\widetilde{x})$ where $-\widetilde{x}$ is the inverse of $\widetilde{x}$ in $\mathbb{Z}_3^{\nu-1}$. Observe that insertion and deletion of $x$ cancel each other out, and that inserting twice the same key $x$ is equivalent to deleting $x$ and vice versa. 
	
	IBLT operations (Figure~\ref{fig:algorithms}) are modified as follows. 
	Operation \texttt{looksPure}  checks if the first trit equals 
	$1$ or $2\equiv -1$. 
	If it is $2$, the key assumed to sit in this entry is the inverse of the retrieved key value. $\mathtt{toggle}(x)$ deletes or inserts $x$ from/to each of the three entries at $h(x)$ depending on whether $x$ has been identified as \textit{pure} with the first trit $1$ or $2$, respectively. Merging two IBLTs $\mathtt{merge}(A,A')$ now becomes non-commutative and is done by subtracting (in $\mathbb{Z}_3^\nu$) $A'$ from $A$ entry-wise. 
	

	Consider a sketch resulting from \texttt{merge}($A,A'$). When reporting a set difference (modified Algorithm \texttt{report} in Figure~\ref{fig:sketch}) the algorithm now outputs \textit{signed elements}, where elements annotated with $1$ (i.e.\ those with the first trit equal to $1$ at the moment of toggling) are interpreted as belonging to the first set, while those annotated with $2$ are interpreted as belonging to the second. The symmetric difference operation $\Delta $ in both \texttt{decode} and \texttt{report} is also modified as follows. When computing $S_1\Delta  S_2$, an element is canceled only if it occurs in $S_1$ and $S_2$ with opposite signs. If an element occurs in $S_1$ and $S_2$ with the same sign, it is reported in the resulting set with the opposite sign. 
	
	Observe that the modified structure does not prevent anomalies to occur but only makes them less likely: an anomaly can still occur if the first trit is $1$ or $2$	but the entry actually contains more than one key. On the other hand,  our analysis of Section~\ref{sec:main_iblt} carries over to the signed case. The difference in $S_{dec}$ compared to the original decoding is due to repeated reportings of the same key with different sign combinations, which does not affect the proof ideas of the main Theorem~\ref{thm:main-inner}. 
}
	
	\subparagraph{Modified control hashsum}  Here the hashsum uses arithmetic summation instead of XOR. We assume that overflow is implemented in a consistent way, that is $(x+y)-y=x$ even if $x+y$ results in overflow. Toggling $x$ entails subtracting $h(x)$ from the hashsum if $x$ is annotated with $1$, and adding if it is annotated $-1$. The hashsum of $\mathtt{diff}(\mathcal{S}_1,\mathcal{S}_2)$ is now defined as $\mathcal{S}_1.H-\mathcal{S}_2.H$. 
	
	\subparagraph{Modified BCH sketch} In order to deal with signed elements, we use BCH code over field $GF(3^{\nu-1})$. The BCH sketch of a set is still defined as the syndrome vector under assumption that each set element is encoded by value $1$ in the error vector. The \texttt{merge\_BCH} operation will now subtract the syndromes (in $GF(3)$) instead of XORing them. By linearity of the code, the ``positive'' and ``negative'' elements will correspond to positions with $1$ and $2\equiv -1$  respectively in the recovered error vector. 
	
	Given a BCH sketch resulting from the difference of sketches of input sets $S$ and $T$, algorithm \texttt{report} will subtract the syndrome corresponding to the set $\Sdec$ of signed elements decoded by the IBLT.  Naturally, positive and negative elements are encoded respectively by values $1$ and $2$ as well. Again, by linearity, the resulting syndrome will encode exactly the elements $\Sdec'$ so that $\Sdec\Delta \Sdec' = S\Delta T$, where positive (resp.\ negative) elements are those belonging to $S$ (resp.\ $T$) only. 
	
	Observe that, in general, $\Sdec$ will contain a subset of $S\Delta T$ as well as possibly some foreign keys, however our modified definition of $\Delta $ will ensure a correct recovery of the original set. As an example, assume $x\in S\setminus T$ and assume $x$ has been reported \gk{by the IBLT} as positive, followed by another reporting of $x$ (produced by an anomaly) as positive as well. From our definition of $\Delta $, $x$ will become negative in $\Sdec$ and will be added (rather than subtracted) to the BCH sketch. This will result in reporting $x$ as negative again by the BCH sketch, and, by our definition of $\Delta $, will be eventually correctly reported as positive in the final output. 
	
	
	The syndrome of the BCH code over $GF(3^{\nu-1})$ consists of at most $2D$ elements of $GF(3^{\nu-1})$ i.e. at most $4D\log_3 U$ bits by a straightforward encoding. 
	Using ternary representation of keys and arithmetic operations in $GF(3^\nu)$ instead of $GF(2^w)$ introduces only a constant factor change in time complexities of both insertion and decoding \gkk{(see Appendix)}. In conclusion, time and robustness guarantees of Theorem~\ref{thm:main} remain valid.

	\section{Concluding remarks}
	
	\gk{
		The presented solution to set reconciliation uses asymptotically negligible additional space and additional time for decoding compared to the IBLT-only solution, but provides a drastic improvement in robustness. 
		The decoding time of our algorithm is small in expectation, however it becomes substantial in worst case. More precisely, when the BCH correction is activated, the decoding time becomes $O(D\,r\,\mathrm{polylog}(U))$ which can approach $O(D^2)$. We believe however this can be  overcome by implementing the stash with expander codes which have much smaller insertion and decoding times. We leave details for future work. 
		
		\gkk{Very recently, we learned about paper \cite{DBLP:journals/corr/abs-2306-07583} that proposes a modified construction of IBLT and applies hash functions of restricted independence, rather than assuming them fully random. As a result, the construction takes a smaller space (compared to the original IBLT with the same error guarantee) and requires less randomness for hash functions. However, space is measured in terms of the number of IBLT buckets rather than in bits, which leads to a much larger multiplicative constant compared to our result. The decoding time of \cite{DBLP:journals/corr/abs-2306-07583} appears to be larger than ours as well (not specified in the paper).}
		}
		
	\bibliography{sketching-biblio.bib}
	\appendix 
	\section*{Appendix}
	Below we describe how to perform multiplication in $GF(2^w)$ and $GF(3^w)$ in RAM model in time  $O(w)$. 
	
	\section{Multiplications in \texorpdfstring{$GF(2^w)$}{GF(2\^{}w)} in RAM model in \texorpdfstring{$O(w)$}{O(w)} time}
	The algorithm we describe is folklore, but since we did not find a reference we briefly sketch it here. The multiplication of two elements $x$ and $y$ proceeds in two phases. In the first phase, one computes the product $z=x\cdot y$, where $\cdot$ refers to a carry-less multiplication of the binary representation of $x$ and $y$. Thus, $z$ will be a bitstring of length $2w-1$. This phase can be easily implemented in $O(w)$ time using $O(w)$ XOR and shift operation\footnote{Note that in practice, most modern processors natively support carry-less multiplication in constant time using much less complex hardware implementation than standard multiplication, but is not part of the standard RAM operations.}. The second phase is the polynomial euclidean division by an irreducible polynomial of degree $w$ and can also be easily implemented in $O(w)$ time using $O(w)$ XOR and shift operations. The remainder of the division will be the final result of the \gkk{multiplication}. 

	\section{Multiplications in \texorpdfstring{$GF(3^w)$}{GF(3\^{}w)} in RAM model  in \texorpdfstring{$O(w)$}{O(w)} time}
	Likewise, multiplication in $GF(3^w)$ can be implemented in $O(w)$ time. For that we can use a representation of an element using $2w$ bits, where each element from the base field $GF(3)$ is represented using $2$ bits. Note that any integer in the range $[0..3^{w}]$ can be converted to this representation in $O(w)$ time, by doing successive euclidean divisions by $3$. Notice that a division by $3$ can be simulated using multiplications and other elementary operations. For a technical reason that will become clear later, we will instead use $4$ bits to represent each element of the base field, resulting in a representation that uses $4w$ bits. Let $G$ be the function that transforms integers of $[0..3^{w}-1]$ into binary representation of $4w$ bits. Let $G^{-1}$ be the inverse of $G$.
	Now assume that given two elements $x$ and $y$ represented as integers, we want to compute $z=x\cdot y$. We will first convert $x$ and $y$ into their binary representations $x'=G(x)$ and $y'=G(y)$, then compute the product $z'=x'\cdot y'$, and finally convert  $z'$ back into an integer $z$ using function $G^{-1}$. 
	We now describe how we do multiplication in the $4w$-bit representation. For that, we successively extract the quadbits ($4$-bit units) of $x'$, multiply each of them with $y'$, and aggregate all of them. Finally, the remainder of euclidean division by an irreducible polynomial can also be done in time $O(w)$ using $O(w)$ operations. 
	Figure~\ref{fig:mult_GF3_algorithm} shows the pseudocode of all steps of the algorithm. 
	
	\begin{figure}
		\begin{minipage}{0.55\textwidth}
			\begin{algorithm}[H]
				\algo{$G(x)$}{
					$x'\leftarrow  0$\;
					$j \leftarrow  0$\;
					\For{$i \leftarrow 1$ \KwTo $w$}{
						$x'\leftarrow  x'  + ((x \mod 3) \ll j)$\;
						$x \leftarrow  x / 3$\;
						$j \leftarrow  j + 4$\;
					}
					\Return $x'$
					
				}
			\end{algorithm}
			\begin{algorithm}[H]
				\algo{$G^{-1}(x')$}{
					$x \leftarrow  0$\;
					$j \leftarrow  4\cdot (w-1)$\;
					\For{$i \leftarrow 1$ \KwTo $w$}{
						$x \leftarrow  x\cdot 3  + ((x' \gg j)\wedge  (0011)_b)$\;
						$j \leftarrow  j - 4$\;
					}
					\Return $x$
					
				}
			\end{algorithm}
			\begin{algorithm}[H]
				\algo{\texttt{MOD3}($x$)}{
					//reduce each quadbit modulo 3\; 
					//works only if quadbit in range [0..5]\; 
						$A \leftarrow  (x+1) \wedge  (0100)_b^w$\;
						$x \leftarrow  x - ((A\gg 2) + (A\gg 1))$\;
						\Return $x$
					}
				\end{algorithm}
			\end{minipage}\hspace{-3em}
			\begin{minipage}{0.55\textwidth}
				\begin{algorithm}[H]
					\algo{$\mathtt{reduce}(x,P)$}{
						//$P$ : irreducible polynomial represented in quadbits\;
						$P' \leftarrow  \mathtt{MOD3}(2\cdot P)$\;
						$j \leftarrow  4\cdot (2\cdot w-1)$\;
						\For{i $\leftarrow $ 1 \KwTo $w$}{
							\If{$(x \gg j) > 0$}{
								$x \leftarrow  x + (0011)_b^w$\;
								\If{$(x \gg j) = (P \gg (w-1))$}{
									$x \leftarrow  x - P$\;
								}
								\Else{
									$x \leftarrow  x - P'$\;
								}
								$x \leftarrow  \mathtt{MOD3}(x)$\;
							}
							$j \leftarrow  j - 4$\;
						}
						\Return $x$
					}
				\end{algorithm}		
				\begin{algorithm}[H]
					\algo{$\mathtt{mult}(x',y')$}{
						$z' \leftarrow  0$\;
						$j \leftarrow  4\cdot (w-1)$\;
						\For{$i \leftarrow  1$ \KwTo $w$}{
							$t \leftarrow   (x' \gg j)\wedge  (0011)_b$\;
							$t \leftarrow  t \cdot y'$\;
							$t \leftarrow  \mathtt{MOD3}(t)$\; 
							$z' \leftarrow  (z'\ll 4)  + t$\;
							$z' \leftarrow  \mathtt{MOD3}(z')$\;
							$j \leftarrow  j - 4$\;
						}
						$z' \leftarrow  \mathtt{reduce}(z',P)$\;
						\Return $z'$
					}
				\end{algorithm}
				\begin{algorithm}[H]
					\algo{$\mathtt{MULT}(x,y)$}{
						$x' \leftarrow  G(x)$\;
						$y' \leftarrow  G(y)$\;
						$z' \leftarrow  \mathtt{mult}(x,y)$\;
						$z \leftarrow  G^{-1}(z)$\;
						\Return $z$				
						
					}
				\end{algorithm}
			\end{minipage}
			\caption{Multiplication algorithm $GF(3^w)$ modulo an irreuducible polynomial}
			\label{fig:mult_GF3_algorithm}
		\end{figure}
		
\end{document}